\documentclass[psamsfonts,reqno]{amsart}

\usepackage{amssymb,amsfonts}
\usepackage[all,arc]{xy}
\usepackage{enumerate}
\usepackage{mathrsfs}
\usepackage{xypic}

\usepackage{amsmath,amssymb,amsfonts,wasysym}   
\usepackage{graphics}                           
\usepackage{amsmath}  
\usepackage{mathtools} 
\usepackage{hyperref}                           
\usepackage{comment}
\usepackage[english]{babel}

\newtheorem{thm}{Theorem}[section]

\newtheorem{prop}[thm]{Proposition}

\theoremstyle{definition}
\newtheorem{defn}[thm]{Definition}

\newtheorem{obs}[thm]{Observation}

\newtheorem{exmp}[thm]{Example}

\theoremstyle{remark}

\newcommand*{\reg}{\mathop{}\!\mathrm{reg}}
\newcommand*{\sing}{\mathop{}\!\mathrm{sing}}

\makeatletter
\let\c@equation\c@thm
\makeatother
\numberwithin{equation}{section}

\title{The maximum likelihood degree of a chemical reaction at the equilibrium}

\author{Simone Camosso}

\date{}

\begin{document}
{\renewcommand{\thefootnote}{\fnsymbol{footnote}}
\setcounter{footnote}{1}
\footnotetext{\textbf{e-mail}: r.camosso@alice.it}
\setcounter{footnote}{0}
}
\begin{abstract}
The complexity of a maximum likelihood estimation is measured by its maximum likelihood degree ($ML$ degree). In this paper we study the maximum likelihood problem associated to chemical networks composed by one single chemical reaction under the equilibrium assumption. 
\end{abstract}

\maketitle

\tableofcontents

\section{Introduction}

The maximum likelihood estimation (MLE) is a method of estimating the parameters of a statistical model given observations. MLE problems appear frequently in experimental sciences. Examples of this diffusion are \cite{AMSW} and \cite{CPR}. In these works the authors consider substances in small concentrations and a discrete random model for chemical reactions (``chemical networks''). Using the maximum likelihood method and numerical tecniques they found an estimation of the rate constants associated to the chemical model. Inspired by these works we offer a new point of view on the topic. 

This article is based on certain analogies between the chemical language and the algebraic statistical formalism. The aim here is modest compared to the works mentioned earlier and,
 in what follows, we will reduce our analysis to consider toy models. To begin, let us list the main differences between our assumptions and those adopted by these authors. First of all our model is not discrete as in \cite{AMSW} and concentrations are not only ``small''. Second, we want to use algebraic statistic methods instead numerical. Third, we assume the initial concentrations know and the work is done in order to determine a theoretical index (the maximum likelihood degree, denoted by $ML$ degree) associated to the chemical kinetics. There is a last assumption that concerns the situation of ``equilibrium'' where the chemical reaction is in a privilegiated condition from the kinetics point of view. The key idea is to interpret the concentration of some chemical substance as a ``frequency'' and we will see how it is possible 
to associate to a chemical process a MLE problem. This formal trick permits to consider a large number of examples. Results are obtained using the methods from the algebraic statistics (as references the reader can consult \cite{HKS}, \cite{HS1}, \cite{HS2}, \cite{DSS}, \cite{HRS}, \cite{M} and \cite{CHKS}) with the auxiliary support of a math software as Maple. Conclusions and considerations are discussed in the last part of the paper. In the preliminaries section necessary math and chemical notations are introduced and explained.

\section{Preliminaries}

\subsection{The Maximum likelihood estimation problem (MLE)}

In algebraic statistic a statistical model is a subset of $\Delta_{n}\,=\,\{p=(p_{0},\ldots,p_{n})\in\mathbb{R}^{n+1}:p_{0},\ldots,p_{n}>0,p_{0}+\ldots +p_{n}\,=\,1\}$ called the 
probability simplex. The real numbers $p_{0},\ldots,p_{n}$ are frequencies and given a statistical model we shall consider the Zariski closure in $\mathbb{P}^{n}$ denoted by $V$ 
as the complex solutions of a system of homogeneous polynomial equations. The maximum likelihood problem consist to find 
$(p_{0},\ldots,p_{n})$ in the model $V_{>0}=V\cap \Delta_{n}$ which ``best explains'' the parameter 
$u=(u_{0},\ldots ,u_{n})\in\mathbb{N}^{n+1}$. This can be obtained maximazing the function:

\begin{equation}
\label{MLformula}
\mathcal{L}_{u}\,=\,\frac{p_{0}^{u_{0}}\cdots p_{n}^{u_{n}}}{(p_{0}+\ldots+p_{n})^{u_{0}+\cdots +u_{n}}},
\end{equation}
\noindent
with the constraint that $p\in V_{>0}$. Let $\lambda\,=\, \sum_{i=0}^{n}u_{i}$ be the dimension of the sample, the problem can be solved using the method of Lagrange multipliers. Furthermore we present another formulation of the same problem in terms more ``algebraic''. Let $\mathcal{H}=\left\{(p_{0},\ldots,p_{n})\in\mathbb{P}^{n}:p_{0}\cdots p_{n}(p_{0}+\ldots+p_{n})=0\right\}$ be the arrangiament of $n+2$ hyperplanes then we are interested for critical points of $\mathcal{L}_{u}$ in  $\mathbb{P}^{n}\setminus \mathcal{H}$. We also restrict our attention for regular points of the model $V_{\reg}=V\setminus V_{\sing}$. We have all elements to define the maximum likelihood degree ($ML$) associated to a statistical model.

\begin{defn}
The maximum likelihood degree $ML$ of $V$ is the number of complex critical points of $\mathcal{L}_{u}$ on $V_{\reg}\setminus \mathcal{H}$, for some $u$.
\end{defn}

In particular for the case of a curve in $\mathbb{P}^{2}$ the following theorem tell us how to calculate the $ML$ degree.

\begin{thm}
\label{huh1}
Let $V$ be a smooth curve of degree $d$ in $\mathbb{P}^2$ and $a=\#(X\cap \mathcal{H})$ the number of points in the distinguished arrangement, then the $ML$ degree of $X$ is $d^2-3d+a$.
\end{thm}

\begin{obs}
We observe that if the curve $V$ is sufficiently generic then $a=4d$ and the $ML$ degree is $d\cdot (d+1)$ as predict by B\'{e}zout theorem for the equations:

$$ f(p_{0},p_{1},p_{2})=0, \ \ \left|\begin{array}{ccc} 1 & 1 & 1 \\ \frac{u_{0}}{p_{0}} & \frac{u_{1}}{p_{1}} & \frac{u_{2}}{p_{2}} \\ \frac{\partial f}{\partial p_{0}} & \frac{\partial f}{\partial p_{1}} & \frac{\partial f}{\partial p_{2}} \end{array} \right|=0.$$

Here $f$ is the homogeneous polynomial generating the curve $V$.
\end{obs}

The previous theorem is a particular case of a general result for very affine varieties (see \cite{H}).

\begin{thm}[Huh]
\label{huh2}
If the very affine variety $X\setminus \mathcal{H}$ is smooth of dimension $d$, then the $ML$ degree is equal to $(-1)^{d}\chi\left(X\setminus \mathcal{H}\right)$.
\end{thm}

\subsection{Chemical reactions}

A chemical reaction is represented by reactants placed to the left of an arrow and products placed to the right. 
We have that both reactants and products are denoted by capital letters $A,B,C, \ldots$. 
The arrow in a chemical reaction can be of three types $\leftarrow$, $\rightarrow$ and $\leftrightarrow$, denoting respectively the direction of evolution of the chemical process. 
The last case is usually used to denote a system that is in equilibrium. Denoting by $A_{1},A_{2}, \ldots, A_{n}$ some reactant and by $B_{1},B_{2},\ldots, B_{m}$ some product, 
we represent the chemical reaction with the following equation:

\begin{equation}
\label{cineticstheory}
\alpha_{1}A_{1}+ \alpha_{2}A_{2} + \cdots + \alpha_{n}A_{n} \rightarrow \beta_{1}B_{1}+\beta_{2}B_{2} + \cdots + \beta_{m}B_{m},
\end{equation}
\noindent
where $\alpha_{1}, \ldots,\alpha_{n},\beta_{1},\ldots,\beta_{m}$ are called stoichiometric coefficients. 
We define the order of a chemical reaction the sum of  $\alpha_{1}, \ldots,\alpha_{n}$. It is possible to define on the set of chemical substances an evaluation map $[\cdot]$ 
that gives the molar concentration of a particular substance represented by the dot $\cdot$. 
Another useful definition in chemical kinetics is the reaction velocity:

\begin{equation}
\label{cineticstheory2}
v\,=\,\frac{1}{\alpha_{i}}\frac{d [A_{i}]}{dt},
\end{equation}
\noindent
where $[A_{i}]$ is the molar concentration of product or reactant $A_{i}$ and $\alpha_{i}$ the stoichiometric coefficient in the reaction. For example the reaction:

$$ I_{2}+Br_{2} \rightarrow 2IBr, $$
\noindent
has a reaction velocity that is the same whether we look at $I_{2}$ ($\alpha=-1$), $Br_{2}$ ($\beta=-1$), or $IBr$ ($\gamma=+2$):

$$ -\frac{d [I_{2}]}{dt}\,=\, -\frac{d [Br_{2}]}{dt} \,=\, \frac{1}{2}\frac{d [IBr]}{dt}.$$

Often, the reaction velocity can be written in terms of a rate law, a power law in the reactant concentrations (or product concentrations), with a concentration--independent coefficient called the (direct) rate constant $K_{d}$:

\begin{equation}
\label{cinetictheory3}
v_{d}\,=\, K_{d} [A_{1}]^{\alpha_{1}}\cdots [A_{n}]^{\alpha_{n}},
\end{equation}
\noindent
or

\begin{equation}
\label{cinetictheory4}
v_{i}\,=\, K_{i} [B_{1}]^{\beta_{1}}\cdots [B_{m}]^{\beta_{m}},
\end{equation}
\noindent 
where the $v_{d}$ and $v_{i}$ stands for ``direct'' and ``inverse'' velocity. For details see \cite{C}.

\subsection{Chemical equilibrium}

There are situations in which both reactants and products are present but have no further tendency to undergo net change, these kind of reactions are called equilibrium reactions. We assume that we are in presence of an equilibrium represented by the following equation:

\begin{equation}
\label{cineticstheor5}
\alpha_{1}A_{1}+ \alpha_{2}A_{2} + \cdots + \alpha_{n}A_{n} \leftrightarrow \beta_{1}B_{1}+\beta_{2}B_{2} + \cdots + \beta_{m}B_{m}.
\end{equation}

Thus we can write the two velocity associated to the kinetic system:

\begin{equation}
\label{cinetictheory6}
v_{d}\,=\, K_{d} [A_{1}]^{\alpha_{1}}\cdots [A_{n}]^{\alpha_{n}},
\end{equation}
\noindent
and

\begin{equation}
\label{cinetictheory7}
v_{i}\,=\, K_{i} [B_{1}]^{\beta_{1}}\cdots [B_{m}]^{\beta_{m}}.
\end{equation}

By the equilibrium assumption we have the equality between $v_{i}$ and $v_{d}$ that can be written as:

$$ K_{i} [B_{1}]^{\beta_{1}}\cdots [B_{m}]^{\beta_{m}}\,=\,K_{d} [A_{1}]^{\alpha_{1}}\cdots [A_{n}]^{\alpha_{n}},$$
\noindent 
and isolating the constants terms we find that:

$$ \frac{K_{i}}{K_{d}}\,=\, \frac{[A_{1}]^{\alpha_{1}}\cdots [A_{n}]^{\alpha_{n}}}{[B_{1}]^{\beta_{1}}\cdots [B_{m}]^{\beta_{m}}}.$$

We denote the term $\frac{K_{i}}{K_{d}}$ by $K_{e}$ and call it the equilibrium constant associated to the reaction $(\ref{cineticstheor5})$. For a more detailed chemical-physical discussion on the subject the reader may refer to \cite{AP}.
         
\section{Results}

\begin{prop}
\label{first}
Let $[A]$ and $[B]$ be the concentrations of certain substances in the following chemical equilibrium reaction of first order:

\begin{equation}
\label{first1}
A \leftrightarrow B,
\end{equation}
\noindent
then the $ML$ degree is equal to $1$ for $K_{e}\,\not=\,-1$ and $0$ for $K_{e}\,=\,-1$.
\end{prop}

\begin{proof}
Let $x$ and $y$ be quantities associated respectively to $[A]$ and $[B]$. This is a line in $\mathbb{P}^{2}$. We must study the $ML$ degree of: 

$$ X=V(-y+K_{e}x).$$

Let $\varphi: \mathbb{C}^{*}\rightarrow \left(\mathbb{C}^{*}\right)^{2}$ be the map that $p_{0}\mapsto (p_{0},K_{e}p_{0})$  with the constraint $p_{0}(1+K_{e})\,=\,1$, then the likelihood--log function is $\mathcal{L}_{u_{0},u_{1}}\,=\, u_{0}\log{p_{0}}+u_{1}\log{K_{e}p_{0}}$. Studing the critical points of the likelihood--log under the constraint we find that:

$$ p_{0}\,=\, \frac{u_{0}+u_{1}}{\lambda(1+K_{e})}.$$

The conclusion follows.
\end{proof}

\begin{prop}
\label{second}
Let $[A],[B]$ and $[C]$ be the concentrations of certain substances in the following chemical equilibrium reaction of second order:

\begin{equation}
\label{first11}
A+B \leftrightarrow 2C,
\end{equation}
\noindent
then the $ML$ degree is $1$ for $K_{e}=4$, $0$ for $K_{e}=0$ and $2$ in the other cases. 
\end{prop}

\begin{proof}
For the equation $(\ref{first11})$ the equilibrium constant is given by:

$$ K_{e}\,=\, \frac{[C]^2}{[A][B]}.$$

Under the assumption that $[A]+[B]+[C]=c$ with $c>0$, calling $x=\frac{[A]}{c},y=\frac{[B]}{c}$ and $z=\frac{[C]}{c}$ we have:

$$ x+y+z=1.$$

The variety of interest is: 

$$ X=V(K_{e}xy-z^2),$$
\noindent 
and the case of $K_{e}=4$ is the Hardy--Weinberg law (details are in \cite{E}, \cite{Ca} and \cite{Ha}). The case $K_{e}=0$ gives as points in the intersection:

$$ X\cap\mathcal{H}\,=\,\{(1:0:0),(0:1:0)\},$$
\noindent 
so the $ML$ degree is $0$. What remain to examine is the case of $K_{e}\not=0,4$. In this case we have the following equations system:

$$\begin{cases} K_{e}xy=z^2 \\ -\frac{2zu_{1}}{y}+\frac{K_{e}yu_{2}}{z}+\frac{u_{0}}{K_{e}}-\frac{u_{1}}{K_{e}}+\frac{2zu_{0}}{x}-\frac{K_{e}xu_{2}}{z}=0 \end{cases}.$$

This leads to two solutions of the following form:

$$\left(\varepsilon_{i}:1:K_{e}\varepsilon_{i}\right),$$
\noindent
where $\varepsilon_{i}$ for $i=1,2$ is a solution of the polynomial equation $2K_{e}^{2}z^2u_{1}-u_{2}K_{e}+(-2K_{e}^{2}u_{0}-u_{0}+u_{1}+u_{2}K_{e})z=0$. This proves that the $ML$ degree is $2$. 
\end{proof}

\begin{prop}
Let $[A]$ and $[B]$ be the concentrations of certain substances in the following chemical equilibrium reaction:

\begin{equation}
\label{first1}
n A \leftrightarrow n B,
\end{equation}
\noindent
for $n\,=\,2,3$ then the $ML$ degree is $1$.
\end{prop}

\begin{proof}
Let $x$ and $y$ be quantities associated respectively to $[A]$ and $[B]$. Let $K_{e}x^n-y^n=0$ be the equation defining $X$. In order to determine the ML degree we consider the map $\varphi: \mathbb{C}^{*}\rightarrow \left(\mathbb{C}^{*}\right)^{2}$ that $p_{0}\mapsto (p_{0},\sqrt[n]{K_{e}}p_{0})$  with the constraint $p_{0}(1+\sqrt[n]{K_{e}})\,=\,1$, then the likelihood--log function is $\mathcal{L}_{u_{0},u_{1}}\,=\, u_{0}\log{p_{0}}+u_{1}\log{\sqrt[n]{K_{e}}p_{0}}$. Studing the critical points of the likelihood--log under the constraint we find that:

$$ p_{0}\,=\, \frac{u_{0}+u_{1}}{\lambda(1+\sqrt[n]{K_{e}})}.$$
\end{proof}

\begin{prop}
Let $[A],[B],[C]$ and $[D]$ be the concentrations of certain substances in the following chemical equilibrium reaction:

\begin{equation}
A+B \leftrightarrow C+D,
\end{equation}
\noindent
then the $ML$ degree is $1$.
\end{prop}

\begin{proof}
By the total conservation of the quantities $[A]+[B]+[C]+[D]=c$ we set $x=\frac{[A]}{c},y=\frac{[B]}{c},z=\frac{[C]}{c}$ and $t=\frac{[D]}{c}$. Our models is the well know independence model of \cite{DSS} \S 1.1 given by $X=V(K_{e}xy-zt) \subset \mathbb{P}^3$. The variety $X$ is isomorphic to $\mathbb{P}^{1}\times \mathbb{P}^{1}$ with coordinates $((K_{e}x:y),(z:t))$. We have that: 

$$ X\setminus\mathcal{H}\,=\, \mathbb{P}^{1}\times \mathbb{P}^{1}\setminus \{K_{e}xyzt(K_{e}x+y)(z+t)=0\}$$
$$\,=\,(\mathbb{P}^{1}\setminus\{K_{e}xy(K_{e}x+y)=0\})\times(\mathbb{P}^{1}\setminus\{zt(z+t)=0\})$$
$$\,=\,(\mathbb{P}^{1}\setminus\{3 \text{hyperplanes}\})\times(\mathbb{P}^{1}\setminus\{3 \text{hyperplanes}\})$$
\noindent
and by theorem $\ref{huh2}$ we find that the $ML$ degree is $\chi(X\setminus\mathcal{H})\,=\, (-1)\cdot (-1)\,=\,1$.
\end{proof}

\begin{obs}
We observe that the equilibrium constant $K_{e}$ is given by the Arrenius formula:

$$ K_{e}\,=\, e^{\frac{\Delta G}{RT}},$$
\noindent 
where $T$ is the temperature, $R$ the gas constant and $\Delta G$ the Gibbs free energy. For this reason it makes sense only consider the case of strictly positive $K_{e}$. 
\end{obs}

\begin{prop}
Let $[A],[B]$ and $[C]$ be the concentrations of certain substances in the following chemical equilibrium reaction:

\begin{equation}
A+B \leftrightarrow 3C,
\end{equation}
\noindent
then the $ML$ degree is $9$.
\end{prop}

\begin{proof}
By the total conservation of the quantities $[A]+[B]+[C]=c$, we set $x=\frac{[A]}{c},y=\frac{[B]}{c}$ and $z=\frac{[C]}{c}$. The variety of interest is: 

$$ X\,=\,V(z^3-K_{e}xy).$$

For convenience we fix $K_{e}=1$. We have the following transformation $\varphi: \left(\mathbb{C}^{*}\right)^2\rightarrow \left(\mathbb{C}^{*}\right)^3$ given by $(p_{0},p_{1})\mapsto \left(p_{0}^{3},p_{1}^{3},p_{0}p_{1}\right)$ with the constraint $p_{0}^{3}+p_{1}^{3}+p_{0}p_{1}\,=\,1$. We study the critical points of the log--likelihood function under the constraint:

\begin{equation}
\label{likelihoodfunction}
\mathcal{L}_{u_{0},u_{1},u_{2}}\,=\, 3u_{0}\log{p_{0}}+3u_{1}\log{p_{1}}+u_{2}\log{p_{0}}+u_{2}\log{p_{1}},
\end{equation}
\noindent
where $u_{0},u_{1},u_{2}$ are a set of parameters. The critical equations are:

\begin{equation}
\label{deff1}
\frac{3u_{0}}{p_{0}}+\frac{u_{2}}{p_{0}}\,=\, \lambda (3p_{0}^2+p_{1}),
\end{equation}
\noindent 
and

\begin{equation}
\label{defg2}
\frac{3u_{1}}{p_{1}}+\frac{u_{2}}{p_{1}}\,=\, \lambda (3p_{1}^2+p_{0}).
\end{equation}

We can denote the polynomial $(\ref{deff1})$ by $f$ and the polynomial $(\ref{defg2})$ by $g$. The Sylvester matrix is: 

$$\text{Syl}(f,g,p_{0})\,=\,\left(\begin{array}{cccc} 3\lambda & 0 & \lambda p_{1} & -3u_{0}-u_{2} \\ \lambda p_{1} & 3\lambda p_{1}^{3}-3u_{1}-u_{2} & 0 & 0\\ 0&\lambda p_{1} &3\lambda p_{1}^{3}-3u_{1}-u_{2} &0 \\ 0 & 0 & \lambda p_{1}& 3\lambda p_{1}^{3}-3u_{1}-u_{2}\end{array} \right),$$
\noindent 
and the resultant $\text{Res}(f,g,p_{0})\,=\,\det{\left(\text{Syl}(f,g,p_{0})\right)}$ is a polynomial of nine degree in $p_{1}$, for the fundamental theorem of algebra we have $9$ solutions.
\end{proof}

\begin{prop}
Let $[A]$ and $[B]$ be the concentrations of certain substances in the following chemical equilibrium reaction:

\begin{equation}
2A \leftrightarrow 3B,
\end{equation}
\noindent
then the $ML$ degree is $3$.
\end{prop}

\begin{proof}
By the total conservation of the quantities $[A]+[B]=c$ we set $x=\frac{[A]}{c}$ and $y=\frac{[B]}{c}$. The variety of interest is: 

$$ X\,=\,V(y^3-K_{e}x^2).$$

For convenience we fix $K_{e}=1$. We have the following transformation $\varphi: \left(\mathbb{C}^{*}\right)\rightarrow \left(\mathbb{C}^{*}\right)^2$ given by $(p_{0})\mapsto \left(p_{0}^{3},p_{0}^{2}\right)$ with the constraint $p_{0}^{3}+p_{0}^{2}\,=\,1$. We study the critical points of the log--likelihood function under the constraint:

\begin{equation}
\label{likelihoodfunction}
\mathcal{L}_{u_{0},u_{1}}\,=\, 3u_{0}\log{p_{0}}+2u_{1}\log{p_{0}},
\end{equation}
\noindent
where $u_{0},u_{1}$ are a set of parameters. The critical equation is:

\begin{equation}
\label{deff}
\frac{3u_{0}}{p_{0}}+\frac{2u_{1}}{p_{0}}\,=\, \lambda (3p_{0}^2+2p_{0}).
\end{equation}

This is a polynomial in $p_{0}$ of third degree and for the fundamental theorem of algebra there are $3$ solutions. 
\end{proof}

\begin{prop}
Let $[A],[B]$ and $[C]$ be the concentrations of certain substances in the following chemical equilibrium reaction:

\begin{equation}
nA+mB \leftrightarrow pC,
\end{equation}
\noindent
then for $n=m=p=2$ the $ML$ degree is $8$, $n=m=2$ and $p=1$ the $ML$ degree is $4$, $n=p=1$ and $m=2$ the $ML$ degree is $2$, $n=m=p=3$ the $ML$ degree is $9$.
\end{prop}

\begin{proof}
For convenience we fix $K_{e}=1$. 
We have the following transformation $\varphi: \left(\mathbb{C}^{*}\right)^2\rightarrow \left(\mathbb{C}^{*}\right)^3$ given by $(t_{0},t_{1})\mapsto \left(t_{0}^{p},t_{1}^{p},t_{0}^{n}t_{1}^{m}\right)$ with the constraint $t_{0}^{p}+t_{1}^{p}+t_{0}^{n}t_{1}^{m}\,=\,1$. We study the critical points of the log--likelihood function under the constraint:

\begin{equation}
\label{likelihoodfunction}
\mathcal{L}_{u_{0},u_{1},u_{2}}\,=\, (pu_{0}+nu_{2})\log{t_{0}}+(pu_{1}+mu_{2})\log{t_{1}},
\end{equation}
\noindent
where $u_{0},u_{1},u_{2}$ are a set of parameters. From the critical equations we find the two polynomial:

\begin{equation}
\label{poly111}
f\,=\,\lambda p t_{0}^{p}+n\lambda t_{0}^{n}t_{1}^{m}-nu_{2}-pu_{0},
\end{equation}
\noindent
and

\begin{equation}
\label{poly222}
g\,=\,\lambda p t_{1}^{p}+m\lambda t_{0}^{n}t_{1}^{m}-mu_{2}-pu_{1}.
\end{equation}

We start considering the case $n=m=p=2$ with $f(t_{0},t_{1})\,=\, 2\lambda t_{0}^2+2\lambda t_{0}^{2}t_{1}^{2} + a$, $g(t_{0},t_{1})\,=\,2\lambda t_{1}^{2}+2\lambda t_{0}^{2}t_{1}^{2}+b$, $a\,=\, -nu_{2}-pu_{0}$ and $b\,=\, -mu_{2}-pu_{1}$. The Sylvester matrix is: 

$$\text{Syl}(f,g,t_{0})\,=\,\left(\begin{array}{cccc} 2\lambda+2\lambda t_{1}^{2} & 0 & a & 0 \\ 0 & 2\lambda+2\lambda t_{1}^{2} & 0 & a\\ 2\lambda t_{1}^{2} & 0 & 2\lambda t_{1}^{2}+b & 0 \\ 0 & 2\lambda t_{1}^{2} & 0 & 2\lambda t_{1}^{2} +b \end{array} \right),$$
\noindent 
with 

$$\text{Res}(f,g,p_{0})\,=\,16\lambda^{4}t_{1}^{4}+16\lambda^3t_{1}^{2}b+4\lambda^2b^2+32\lambda^4 t_{1}^{6}+32\lambda^3t_{1}^{4}b+8\lambda^2t_{1}^{2}b^2-16\lambda^3t_{1}^{4}a-$$

$$-8\lambda^2 t_{1}^{2}ab+16\lambda^4t_{1}^{8}+16\lambda^3t_{1}^{6}b+4\lambda^2t_{1}^{4}b^2-16\lambda^3t_{1}^6a-8\lambda^2t_{1}^{4}ab+4\lambda^2t_{1}^{4}a^2,$$ 
\noindent 
a polynomial of eight degree in $p_{1}$ and for the fundamental theorem of algebra we have $8$ solutions.

For the case $n=m=2$ and $p=1$ with $f(t_{0},t_{1})\,=\, \lambda t_{0}+2\lambda t_{0}^{2}t_{1}^{2} + a$, $g(t_{0},t_{1})\,=\,\lambda  t_{1}+2\lambda t_{0}^{2}t_{1}^{2}+b$, $a\,=\, -nu_{2}-pu_{0}$ and $b\,=\, -mu_{2}-pu_{1}$, the Sylvester matrix is: 

$$\text{Syl}(f,g,t_{0})\,=\,\left(\begin{array}{cccc} 2\lambda t_{1}^{2} & \lambda & a & 0 \\ 0 & 2\lambda t_{1}^{2} & \lambda & a\\ 2\lambda t_{1}^{2}&0 & \lambda t_{1}+b &0 \\ 0 & 2\lambda t_{1}^{2} & 0& \lambda t_{1}+b\end{array} \right),$$
\noindent 
with 

$$\text{Res}(f,g,p_{0})\,=\,4\lambda^{4}t_{1}^{6}+8\lambda^3t_{1}^{5}b+4\lambda^2t_{1}^{4}b^2-8\lambda^3 t_{1}^{5}a-8\lambda^2t_{1}^{4}ab+2\lambda^4t_{1}^{3}+2\lambda^3t_{1}^{2}b+4\lambda^2t_{1}^{4}a^2,$$
\noindent 
a polynomial of six degree in $p_{1}$ and for the fundamental theorem of algebra we have $6$ solutions but with only $4$ different from zero.

In the case $n=p=1$ and $m=2$ with $f(t_{0},t_{1})\,=\, \lambda t_{0}+\lambda t_{0}t_{1}^{2} + a$, $g(t_{0},t_{1})\,=\,\lambda  t_{1}+2\lambda t_{0}^{2}t_{1}^2+b$, $a\,=\, -nu_{2}-pu_{0}$ and $b\,=\, -mu_{2}-pu_{1}$, the Sylvester matrix is: 

$$\text{Syl}(f,g,t_{0})\,=\,\left(\begin{array}{cc} \lambda +\lambda t_{1}^{2} & a  \\ 2\lambda t_{1}^{2} & \lambda t_{1}+b  \end{array} \right),$$
\noindent 
with 
$$\text{Res}(f,g,p_{0})\,=\, \lambda^2t_{1}+\lambda b+\lambda^2t_{1}^{3}+\lambda t_{1}^{2}b-2a\lambda t_{1}^{2},$$
\noindent 
a polynomial with only $2$ solutions different from zero.

In the last case $n=m=p=3$ with $f(t_{0},t_{1})\,=\, 3\lambda t_{0}^3+3\lambda t_{0}^{3}t_{1}^{3} + a$, $g(t_{0},t_{1})\,=\,3\lambda  t_{1}^{3}+3\lambda t_{0}^{3}t_{1}^{3}+b$, $a\,=\, -nu_{2}-pu_{0}$ and $b\,=\, -mu_{2}-pu_{1}$, the Sylvester matrix is: 

$$\text{Syl}(f,g,t_{0})\,=\,\left(\begin{array}{cccccc} 3\lambda +3\lambda t_{1}^{3} & 0 & 0 & a & 0 & 0  \\ 0 & 3\lambda +3\lambda t_{1}^{3} & 0 & 0 & a & 0 \\ 0 & 0 & 3\lambda +3\lambda t_{1}^{3} & 0 & 0 & a \\ 3\lambda t_{1}^{3} & 0 & 0 & 3\lambda t_{1}^{3} + b & 0 & 0 \\ 0 & 3\lambda t_{1}^{3} & 0 & 0 & 3\lambda t_{1}^3 + b & 0 \\ 0 & 0 & 3\lambda t_{1}^{3} & 0 & 0 & 3\lambda t_{1}^{3}+b  \end{array} \right),$$
\noindent 
with 
$$\text{Res}(f,g,p_{0})\,=\, \left(9\lambda^2t_{1}^{3}+3\lambda b+9\lambda^2t_{1}^{6}+3\lambda t_{1}^{3} b -3a\lambda t_{1}^{3}\right)^3,$$
\noindent 
a polynomial with $9$ solutions.
\end{proof}

\begin{prop}
Let $[A_{1}],[A_{2}], \ldots, [A_{n}]$ and $[B_{1}],[B_{2}],\ldots,[B_{n}]$ be the concentrations of certain substances in the following chemical equilibrium reaction:

\begin{equation}
A_{1}+A_{2}+\cdots +A_{n} \leftrightarrow B_{1}+B_{2}+\cdots +B_{n},
\end{equation}
\noindent
then the $ML$ degree is $1$.
\end{prop}

\begin{proof}
We observe that the number or reactants is equal to the number of products that is $n$. We consider the following transformation $\varphi: \left(\mathbb{C}^{*}\right)\rightarrow \left(\mathbb{C}^{*}\right)^{2n}$ given by 

$$t_{0}\mapsto \left(t_{0}, \ldots ,t_{0}, \sqrt[n]{K_{e}}t_{0}, \ldots, \sqrt[n]{K_{e}}t_{0}\right).$$

Proceeding in a similar way as other results we find that the $ML$ must be $1$.
\end{proof}

\begin{exmp}
As example we can consider the synthesis of ammonia at the pressure of $800\,atm$ and at $T=500^{\circ}\,C$. At the equilibrium:

$$ N_{2} + 3H_{2} \leftrightarrow 2NH_{3}. $$

The transformation map $\varphi$ is given $\left(t_{0},t_{1}\right)\mapsto \left(t_{0}^{2},t_{1}^{2},\sqrt{K_{e}}t_{0}t_{1}^3\right)$. As in the proof of previous results we consider the likelihood--log function:

$$\mathcal{L}_{u_{0},u_{1},u_{2}}\,=\, \left(2u_{0}+u_{2}\right)\log{t_{0}}+\left(2u_{1}+3u_{2}\right)\log{t_{1}}+u_{2}\log{\sqrt{K_{e}}},$$
\noindent 
with the constraint $t_{0}^{2}+t_{1}^{2}+\sqrt{K_{e}}t_{0}t_{1}^3\,=\, 1$. The procedure leads to the determinant of the Sylvester matrix to be a polynomial of degree $8$ with the numeric coefficient different from $0$. In this example the $ML$ degree is equal to $8$. 
\end{exmp}

\section{Conclusions}

In the previous results the interpretation of chemical concentrations as ``frequencies'' leads to different examples of $ML$ degree problems. 
In each example a solution has been proposed. The propositions provide a partial classification of certain chemical reactions by its $ML$ degree and we 
can observe qualitatively the growth of the $ML$ degree to varying complexity of chemical reactions. The fact that no higher order reaction has been considered is due principally by the motivation that reactions with high molecularity are ``rare'' because the probability of effective collision between particles decreases. Another interesting study regards chemical reactions with half order or with no a ``perfect'' equilibrium, in adjoint we don't know how to treat the case when a reaction is composed by more steps in order to arrive to the final products. 

In other words how to treat the case of chemical networks? In \cite{CPR} they introduced the multinomial model. 
The method used here works well only under the equilibrium assumption and it is not possible to use it for general chemical networks. 

In conclusion what emerges on this study is that the $ML$ problems are generally connected to the problem of solving polynomial equations in order to find projective points. 
It is interesting that chemical reactions of high order seems rare in the same way as to find solutions of higher degree equations is not quite obvious (we refer to the Galois famous result on the solvability by radicals). In fact the $ML$ degree is the degree of the extension $\mathbb{K}/\mathbb{Q}(u)$ obtained adjoining all solutions of the likelihood equations to $\mathbb{Q}(u)$. In this notation $\mathbb{Q}(u)$ is the field of rational functions and $u$ is the indeterminate vector of parameters $u=\left(u_{0},\ldots ,u_{n}\right)$ as observed by \cite{HRS} \S 4.


\begin{thebibliography}{9}

\bibitem{AMSW}
A.Andreychenko, L.Mikeev, D.Spieler, V.Wolf, ``Approximate maximum likelihood estimation for stochastic chemical kinetics'', EURASIP J. Bioinformatics and Syst Biology, (2012),  2012(1): 9.

\bibitem{AP}
P.Atkins, J.de Paula, ``Atkins'Physical Chemistry'', W.H. Freeman and Company New York (2006), Eighth edition, pp 200--202. 

\bibitem{Ca}
S.Camosso, ``Considerations on the genetic equilibrium law'' , IOSR Journal of Mathematics (IOSR-JM),  Volume 13, Issue 1, Ver. I (Jan.-Feb. 2017), pp 01--03.

\bibitem{C}
A.Cooksy, ``Physical Chemistry: Thermodynamics, Statistical Mechanics, \& Kinetics'', Pearson (2014). 

\bibitem{CPR}
G.Craciun, C.Pantea, G.Rempala (2009), ``Algebraic methods for inferring biochemical networks: a maximum likelihood approach'', Comput. Biol. Chem. 33(5), pp. 361–367.

\bibitem{CHKS}
 F.Catanese, S.Hoşten, A.Khetan, B.Sturmfels, ``The maximum likelihood degree'', Amer. J. Math. 128 (2006), no. 3, 671--697. MR 2230921.
 
\bibitem{DSS}
M.Drton, B.Sturmfels, S.Sullivant, ``Lecture on Algebraic Statistics'', Oberwolfach Seminars, Vol. 39, Kirkh\"{a}user, Basel (2009).

\bibitem{E}
A.W.F.Edwards, ``Foundations of Mathematical Genetics'', Cambridge University Press, Cambridge (2000).

\bibitem{Ha} 
G.H.Hardy, ``Mendelian Proportions in a Mixed Population'', Science, New Series, Vol.28, 706 (1908), 49--50.

\bibitem{H} 
J.Huh, ``The maximum likelihood degree of a very affine variety'', Compositio Math., 149, 1245--1266 (2013).

\bibitem{HKS} S.Ho\c{s}ten, A.Khetan, B.Sturmfels, ``Solving the Likelihood Equations'', Foundations of Computational Mathematics, Vol. 5, Issue 4, pp 389--407, 2005.

\bibitem{HRS} J.Hauenstein, J.Rodriguez, B.Sturmfels, ``Maximum Likelihood for Matrices with Rank Constraints'', Journal of Algebraic Statistics, Vol.5, Issue 1 (2014), pp 18--38.

\bibitem{HS1} S.Ho\c{s}ten, S.Ruffa, ``Introductory Notes to Algebraic Statistics'', Rend. Istit. Mat. Univ. Trieste, Vol. XXXVII, 39--70 (2005).

\bibitem{HS2} 
J.Huh, B.Sturmfels, ``Likelihood Geometry'', Combinatorial Algebraic Geometry: Levico Terme, Italy 2013, Springer International Publishing, Vol.2108 of the series Lecture Notes in Mathematics (2014), 63--117.

\bibitem{M}
I.J.Myung, ``Tutorial on maximum likelihood estimation'', Journal of Mathematical Psychology 47 (2003) 90--100.

\end{thebibliography}
\end{document}